\documentclass[10pt,journal,twoside]{IEEEtran}
\usepackage{comment}
\usepackage{amsmath,amssymb,amsthm,mathtools}
\usepackage{graphicx,booktabs}
\usepackage[T1]{fontenc}
\usepackage[expansion=false,protrusion=false]{microtype}
\usepackage[dvipsnames]{xcolor}
\definecolor{linkc}{RGB}{20,60,120}
\usepackage[colorlinks=true,allcolors=linkc]{hyperref}
\usepackage[capitalise]{cleveref}

\usepackage[numbers,sort&compress]{natbib}

\newtheorem{theorem}{Theorem}
\newtheorem{proposition}[theorem]{Proposition}
\newtheorem{lemma}[theorem]{Lemma}
\newtheorem{corollary}[theorem]{Corollary}
\theoremstyle{definition}
\newtheorem{definition}[theorem]{Definition}
\newtheorem{remark}[theorem]{Remark}

\DeclareMathOperator{\rank}{rank}

\DeclareMathOperator{\intr}{int}

\newcommand{\R}{\mathbb{R}}
\newcommand{\Z}{\mathbb{Z}}

\newcommand{\Inc}{\mathsf{B}}
\newcommand{\Cyc}{\mathsf{C}}
\newcommand{\Vpot}{U}

\title{When Gigawatts of Computational Load Disappear: Cycle-Space Certificates for Grid Synchronization and Transient Stability}
\author{Michael (Misha) Chertkov
\thanks{M. Chertkov is with the Program in Applied Mathematics and the Department of Mathematics, University of Arizona, Tucson, AZ 85721 USA (e-mail: chertkov@arizona.edu).}}

\begin{document}
\maketitle

\begin{abstract}
Rapid growth of data centers and artificial-intelligence services is producing computational loads at scales once associated mainly with largest power plants. Recent grid events show that a routine transmission disturbance can cause several gigawatts of data-center demand to disconnect or transfer to backup nearly at once. This article revisits the classical synchronization and transient-stability theory needed to reason about such events. We organize four lines of work---graph-based synchronization conditions, winding-number descriptions of nonlinear power flow, separable convex network optimization, and direct energy methods---into a single cycle-space certificate framework for the lossless fixed-voltage model. The static layer gives an exact strict-cohesion test within a prescribed winding cell and reveals the widely used D\"orfler--Chertkov--Bullo test as a quadratic surrogate of the same convex problem. The dynamic layer converts the critical-energy calculation into a finite family of convex boundary problems. Standard MATPOWER benchmarks illustrate both what the stronger static test gains and where it gains nothing: the 118-bus case admits $16.2\%$ more loading than the sufficient screen, while the 39-bus case is bridge-limited and the thresholds coincide. A stylized $2.7$-GW 39-bus event further shows that transient margin can change by about a factor of two depending on where balancing power is supplied, even when every final balanced operating point remains statically feasible. The result is a tutorial synthesis and an extensible deterministic certificate for emerging gigawatt-scale computational-load contingencies.
\end{abstract}

\begin{IEEEkeywords}
Computational loads, convex optimization, cycle space, data centers, power-system synchronization, transient stability.
\end{IEEEkeywords}

\section{Why Computational Loads Reopen a Classical Synchronization Problem}
\label{sec:intro}

This is no longer a hypothetical class of contingency.  On July~22, 2026, a mechanical failure removed a 230-kV line from service in Northern Virginia.  The fault itself was correctly cleared, but data centers in the Dominion zone then disconnected unexpectedly and transferred to on-site backup generation.  PJM's August~11 review reports an initial load drop of $2970$~MW followed by a second $1099$~MW wave triggered by the resulting voltage disturbance; in total, approximately $3.8$~GW of data-center load disappeared from the grid without warning, the largest such event PJM had experienced \cite{pjm2026largeload}.  The earlier technical presentation reports an Area Control Error excursion to $+3928$~MW and frequency reaching $60.092$~Hz \cite{pjm2026dominion}.  The event is therefore not evidence that artificial-intelligence training \emph{initiated} the transmission fault.  It is evidence of a different and more directly relevant fact: computational facilities can form a geographically concentrated, multi-gigawatt load whose protection and ride-through behavior can transform an otherwise cleared grid fault into an abrupt bulk-power imbalance.

The reliability community is responding accordingly.  PJM is evaluating ride-through requirements for large computational loads, and in July~2026 the Federal Energy Regulatory Commission directed the North American Electric Reliability Corporation to develop enforceable reliability standards and registration criteria for computational-load integration \cite{ferc2026computational}.  Large intermittent computational loads were already recognized as a voltage-control challenge at the high-performance-computing scale a decade ago \cite{zhao2014intermittent}; what has changed is the scale, concentration, and speed with which the load can now move.

This concern is broader than a single PJM event.  NERC's Large Loads Action Plan identifies computational facilities---including data centers, cryptocurrency facilities and artificial-intelligence workloads---as a distinct reliability problem because of their rapid growth, clustering, uncertain models and disturbance behavior \cite{nerc2026largeloads}.  NERC's 2026 risk-mitigation guidance specifically discusses coordinated voltage and frequency ride-through, dynamic reactive support, and fast downward-responsive resources as possible responses to sudden loss of large load \cite{nerc2026risk}.  Those recommendations span several layers of physics.  The present article addresses only one of them: the network-angle geometry that determines whether active-power redistribution admits a cohesive synchronous state and whether an electromechanical trajectory can be certified to reach it.

This change of scale has motivated recent work on planning, scheduling and coordinating data-center integration with the host grid.  Examples include quantifying the grid-dependent value of shifting or curtailing data-center demand \cite{khanal2026shift} and runtime co-simulation of data-center/grid coordination using measured AI-service workloads and grid simulators \cite{chung2026openg2g}.  The question taken up here is complementary and deliberately narrower: {\bf if a large load changes abruptly and balancing has been provisioned, can one certify that the transmission-grid angle dynamics remain synchronized?}

Why does this become a synchronization problem rather than only a generation-balancing problem?  When a large load disappears, total generation initially exceeds total demand and system frequency rises.  But the imbalance is also spatial: injections change at particular buses, power reroutes through the transmission graph, and relative voltage angles must move to support the new line-flow pattern.  A system can therefore possess enough aggregate balancing capability and still encounter an angle-security limitation.  This paper isolates that electromechanical geometry.  It asks first {\bf whether the prescribed post-event injections admit a cohesive synchronous state} and then {\bf whether the swing trajectory from the pre-event operating point is guaranteed to reach that state without crossing the cohesive boundary}.

The engineering question therefore has two halves that are easy to conflate.  The first is \emph{static}: for a specified balanced injection vector, does a strictly cohesive synchronous operating point exist?  D\"orfler, Chertkov and Bullo \cite{dorfler2013synchronization} give a sharp and widely used sufficient synchronization condition.  Exact characterizations based on winding cells and algorithms guaranteed to recover all solutions are available in the later flow/elastic-network framework of Jafarpour, Huang, Smith and Bullo \cite{jafarpour2022torus}.  Sections~\ref{sec:dual} and~\ref{sec:improve} expose a particularly simple separable convex-cost realization for the sinusoidal lossless model and identify precisely where the D\"orfler--Chertkov--Bullo screen sits relative to that exact cell-wise test.

The second question is \emph{dynamic}: starting from the old angles and applying the new balanced injection vector instantaneously, does the swing trajectory stay inside the cohesive region and converge to the new equilibrium?  \Cref{sec:trans} gives an energy certificate, and \Cref{sec:num} compares it with finite-horizon simulations.  The two questions share the same convex potential on the cohesive region, but they should not be conflated: static feasibility can hold while the transient certificate fails.

The organization of the certificate can be summarized as a sequence of questions rather than as a sequence of algorithms:

\begin{figure*}[t]
\centering
\fbox{\begin{minipage}{0.94\textwidth}
\centering
\textbf{disturbance and prescribed balancing allocation}
$\Longrightarrow$
\textbf{post-event injections}
$\Longrightarrow$
\textbf{cycle-space equilibrium oracle}
$\Longrightarrow$
\textbf{cohesive equilibrium}
$\Longrightarrow$
\textbf{boundary energy}
$\Longrightarrow$
\textbf{deterministic transient certificate}
\end{minipage}
}
\caption{The static-to-dynamic certificate pipeline used in this article.  The balancing allocation is an input to the calculation, not a controller synthesized by it.  Static feasibility and transient security are distinct layers: a post-event equilibrium can exist even when the energy certificate does not guarantee that the pre-event state reaches it without leaving the cohesive region.}
\label{fig:pipeline}
\end{figure*}

\paragraph{What is deliberately left outside the model.}
A sudden loss or transfer of a large load excites several coupled layers of grid physics; this paper isolates only the electromechanical angle layer.  Voltage and reactive-power excursions are not secondary in importance.  Indeed, the July~22 PJM event itself included a high-voltage excursion that triggered additional load transfer \cite{pjm2026dominion}.  Prior work on intermittent high-performance-computing demand showed that fast local reactive support, for example a D-STATCOM coordinated with slower switched capacitors, can substantially improve voltage regulation \cite{zhao2014intermittent}.  We therefore bracket voltage dynamics here not because they are automatically benign, nor because they are always confined to milliseconds, but because they require an AC/reactive-power model absent from the fixed-voltage reduction.  Their fastest electromagnetic and converter-control components occur on shorter time scales than the seconds-to-tens-of-seconds electromechanical window emphasized below, while slower voltage recovery can overlap that window.

We likewise do not model the \emph{secondary} consequences of the angle transient, where ``secondary'' refers to causal ordering rather than severity.  Transient line or transformer overloads can activate protection and initiate further outages.  Recent dynamic $N\!-\!1$ screening work, for example, studies overload risk on lines and transformers after common short-circuit faults over sub-second to tens-of-seconds horizons \cite{almada2026dynamic}.  Such overload/protection/cascade mechanisms are important extensions, but the present question stops one layer earlier: does the rotor-angle trajectory itself stay inside the cohesive region and converge?

\paragraph{Review lens and contribution.}
Korsak gave an explicit counterexample to uniqueness of stable load-flow solutions \cite{korsak1972}.  Tavora and Smith studied the flow map, singular surfaces and large-angle stability geometry \cite{tavora1972}.  Winding/loop-flow interpretations and single-loop multistability were developed in \cite{coletta2016topologically,delabays2016multistability,manik2017cycle}.  Most importantly for our claims, Jafarpour, Huang, Smith and Bullo introduced a winding partition of the $n$-torus, proved that each winding cell contains at most one solution under the relevant monotonicity assumptions, bounded the number of solutions, and gave an algorithm guaranteed to compute all of them \cite{jafarpour2022torus}.  Thus \emph{at-most-one solution per winding cell} and \emph{all-solutions enumeration} are prior art.

There is also relevant prior for the convex flow formulation.  Bent, Bienstock and Chertkov \cite{bent2013syncaware} formulated the voltage-uniform lossless synchronization constraint as a separable convex optimization in line-flow variables, with objective obtained by integrating $\arcsin$ and with infeasibility detected when the optimizer reaches the capacity boundary.  This is naturally a monotropic/separable convex-cost network-flow problem in the sense of Rockafellar \cite{rockafellar1984network}.

We build on these foundations and: (i) eliminate flow conservation explicitly with an integer cycle basis and write the resulting strictly convex objective and its winding tilt in cycle coordinates; (ii) show, in the zero-winding cell selected by continuation from the small-angle branch, that the D\"orfler--Chertkov--Bullo (DCB) \cite{dorfler2013synchronization} sufficient synchronization condition is the capacity-box test evaluated at the minimizer of the \emph{quadratic Taylor surrogate} of the same separable convex cost; (iii) identify the bridge case as one in which cycle corrections are impossible; and (iv) use primal convexity on the closed cohesive polytope to decompose the direct-method critical energy into at most $2L$ convex face problems, one for each signed energized-edge constraint.  The numerical sections test these statements on stylized MATPOWER reductions.

For a review-oriented reader, the central point is not that any one of these ingredients is individually new.  It is that they become substantially more useful when placed in one chain of reasoning.  The cycle formulation shows what the classical sufficient screen neglects, the bridge corollary predicts when the stronger nonlinear calculation cannot help, and the same cohesive geometry turns the direct transient-energy barrier into a finite family of convex problems.

\paragraph{Results at a glance.}
The quantitative distinctions are visible on standard benchmark networks.  On MATPOWER case118, the D\"orfler--Chertkov--Bullo sufficient test reaches a loading multiplier $4.363697$, whereas the exact strict-cohesion oracle in the zero-winding cell reaches $5.068736$; a warm-started locally stable continuation persists to approximately $5.268367$.  Thus the convex cohesion oracle permits $16.16\%$ more loading than the quadratic sufficient screen and closes $77.93\%$ of the interval from that screen to the numerically continued branch limit.  On case39 the sufficient and exact cohesive thresholds coincide because the limiting branch is a bridge, on which cycle-space redistribution cannot change the flow.  In the transient calculation, convexity reduces the direct-method critical energy to at most $2L$ signed branch-face problems, with $L=|E|$ the number of energized branches.  The case39 experiments then show a separate effect: static feasibility may remain unchanged while the certified and simulated transient margins vary strongly with the spatial allocation of balancing power.

\section{A Tutorial Map of the Classical Tools}
\label{sec:tutorial}

The synchronization literature contains several mathematically related objects that answer different engineering questions.  For the present application, four are especially useful: a Laplacian-based sufficient synchronization screen, winding numbers that classify nonlinear equilibria on meshed graphs, a separable convex flow formulation of the cohesive power-flow equations, and direct energy methods for transient stability.  The purpose of this section is to place them on one map before using them technically.

\subsection{Three notions that should not be conflated}

\emph{Equilibrium existence} asks whether the post-event injections can be supported by any phase configuration satisfying the nonlinear active-power equations.  \emph{Local stability} asks whether a particular equilibrium is stable to sufficiently small perturbations; for the lossless swing model this is governed by the weighted Laplacian obtained from the cosine of the equilibrium angle differences.  \emph{Transient security} is stronger: starting from the actual pre-event state, does the finite disturbance trajectory remain in a prescribed safe region and converge to the desired equilibrium?  A system may pass the first two questions and fail the third.  Conversely, a conservative transient certificate may fail even when simulation indicates that the trajectory remains stable, because the certificate deliberately proves invariance of a restricted region rather than locating the full nonlinear basin of attraction.

The cohesive region $|\theta_i-\theta_j|<\pi/2$ is particularly useful because it is simultaneously an engineering angle-security surrogate and a mathematical convexity domain.  Inside it, every line contributes a positive stiffness $K_{ij}\cos(\theta_i-\theta_j)$, the swing potential is convex modulo a uniform phase, and the inverse line-flow map uses the principal branch of $\arcsin$.  These three facts are the common thread connecting the static and dynamic constructions below.

\subsection{How the four viewpoints developed}

The first viewpoint is the Laplacian or ``DC-like'' one.  Near a small-angle operating point, the sine coupling is linear and the active-power equations reduce to a weighted graph Laplacian.  This approximation is attractive because one pseudoinverse solve produces all edge angle differences.  D\"orfler, Chertkov and Bullo showed that the same Laplacian object can be used in a nonlinear synchronization theorem: if the edgewise differences of $L_K^\dagger p$ stay below a sine-transformed threshold, then a unique cohesive stable solution exists \cite{dorfler2013synchronization}.  The result is much stronger than a heuristic DC power-flow check because the conclusion concerns the nonlinear sine-coupled system.  What remains hidden in the usual statement is that the tested flow itself is the minimizer of a quadratic network energy.

The second viewpoint is topological.  On a tree, line flows are fixed by conservation once the injections are given.  On a meshed network, however, conservative flows can differ by circulations around cycles.  Because phases are angular variables, principal phase differences around a cycle need not sum to zero as real numbers; they can sum to an integer multiple of $2\pi$.  This observation explains why the same injection vector may support several nonlinear equilibria.  Work on topological states, multistability and winding partitions made this structure explicit \cite{coletta2016topologically,delabays2016multistability,manik2017cycle,jafarpour2022torus}.  Jafarpour \emph{et al.} went further: under the relevant monotonicity conditions, each winding cell contains at most one solution and all solutions can be enumerated by searching the finite winding partition \cite{jafarpour2022torus}.  For the present article, winding numbers are not an alternative power-flow formalism; they are the labels that make the nonlinear cycle constraints unambiguous.

The third viewpoint is variational.  If a cohesive line flow $F_e$ is written as $K_e\sin\Delta_e$, then the principal angle difference is $\Delta_e=\arcsin(F_e/K_e)$.  Integrating this inverse relation produces a convex line cost.  Bent, Bienstock and Chertkov used exactly this structure to formulate lossless synchronization as a separable convex flow problem \cite{bent2013syncaware}.  In network-optimization language it is a monotropic problem \cite{rockafellar1984network}.  The cycle-space formulation used here does not replace that result; it eliminates the conservation constraints explicitly and reveals that winding cells correspond to linear tilts of one strictly convex objective.  This makes the connection to the Laplacian screen transparent: the DC flow is simply what results when the convex line cost is truncated after its quadratic term.

The fourth viewpoint is dynamical and predates much of the modern graph language.  Direct transient-stability methods construct a Lyapunov or energy function for the post-disturbance swing equations and compare the initial post-fault energy with an estimate of the energy required to reach the stability boundary \cite{athay1979practical,chiang1995direct}.  The classical challenge is not writing the energy function; it is locating the relevant boundary energy in a high-dimensional nonconvex landscape, often through a controlling unstable equilibrium point.  The cohesive restriction changes this geometry.  On the closed region where every energized-line difference lies between $-\pi/2$ and $\pi/2$, the same potential is convex.  The boundary is a finite union of signed edge faces, so its minimum energy can be obtained by convex minimization face by face.  This is the second place where the cohesive restriction buys tractability: first in equilibrium existence, then in transient certification.

Taken together, the four viewpoints form an approximation-and-escalation hierarchy rather than four competing theories.  The Laplacian screen is the cheapest and most conservative static layer; the full cycle-space convex program removes its quadratic-flow approximation inside a selected winding cell; the energy calculation adds trajectory information; and nonlinear time simulation remains the natural escalation when the certificate is inconclusive.  A failed sufficient test is therefore never interpreted as proof of instability.  This distinction is important operationally and will reappear in the benchmark discussion.

\subsection{Four complementary viewpoints}

\begin{table*}[t]
\centering
\caption{Four classical viewpoints used in the article.  The synthesis is useful because each viewpoint exposes a different part of the same synchronization problem.}
\label{tab:map}
\small
\begin{tabular}{p{0.17\textwidth}p{0.21\textwidth}p{0.235\textwidth}p{0.235\textwidth}}
\toprule
viewpoint & basic object & what it certifies or explains & principal limitation \\
\midrule
Laplacian synchronization screen \cite{dorfler2013synchronization}
& weighted graph Laplacian and its pseudoinverse
& a computationally cheap sufficient bound for cohesive synchronization
& evaluates the capacity test at the quadratic/DC surrogate flow rather than the full nonlinear cycle-space optimum \\
Winding-cell geometry \cite{coletta2016topologically,delabays2016multistability,manik2017cycle,jafarpour2022torus}
& integer phase advance around graph cycles
& classifies multiple nonlinear equilibria and gives at-most-one solution per winding cell under the relevant monotonicity assumptions
& all cells must still be considered if the goal is to enumerate every cohesive equilibrium \\
Separable convex network flow \cite{bent2013syncaware,rockafellar1984network}
& integrated $\arcsin$ line cost under flow conservation
& exact cohesive feasibility in a selected winding cell; cycle coordinates remove conservation constraints explicitly
& exact only for the lossless fixed-voltage cohesive model used to define the convex domain \\
Direct transient-energy method \cite{athay1979practical,chiang1995direct}
& kinetic plus post-event potential energy
& an invariant energy sublevel contained in the region of attraction
& conservative; classical implementations often require a difficult search for the controlling boundary object \\
\bottomrule
\end{tabular}
\end{table*}

The four complementary viewpoints are summarized in Table \ref{tab:map}.

\subsection{Why cycle space is the organizing coordinate}

Flow conservation and angle consistency are the two structural constraints of the lossless power-flow equations.  Flow conservation is linear: all feasible line-flow vectors form an affine space.  The nullspace of the incidence matrix is precisely the graph cycle space, so every conservative flow differs from a particular one only by circulations around cycles.  Angle consistency is nonlinear because line flows are mapped back to phase differences through $\arcsin$.  Writing the problem in cycle coordinates therefore separates the easy constraint from the nonlinear one: conservation is eliminated once and for all, while winding numbers appear as integer conditions on cycle sums of principal phase differences.

This same decomposition clarifies two limiting cases before any computation is performed.  On a tree there is no cycle space, so the DC/quadratic flow and the exact nonlinear cohesive flow coincide.  On a meshed network the nonlinear optimizer can redistribute stress through cycle flows, but it cannot alter the flow on a bridge because a bridge belongs to no cycle.  These observations become formal corollaries in Section~\ref{sec:improve} and explain the opposite behavior of the 39-bus and 118-bus benchmarks.

\section{Model and Cohesive Geometry}\label{sec:setting}

We now specialize the tutorial map to the standard lossless fixed-voltage network model.  The model suppresses reactive-power and voltage-magnitude dynamics so that active-power transfer is controlled only by phase differences.  It is not intended as a complete model of a modern data center or of a protection event; its value here is that the synchronization geometry is explicit enough to admit exact cell-wise statements and reproducible benchmark calculations.

Let $G=(V,E)$ be a connected graph, $|V|=n$, $|E|=L$. Fix an orientation and let
$\Inc\in\R^{n\times L}$ be the incidence matrix of the graph, $\Inc_{ie}=+1$ if
$i$ is the tail of $e$, $-1$ if the head. Line couplings $K_e>0$ collect into $K$,
injections and consumptions $p\in\R^n$ satisfy $\mathbf 1^\top p=0$, and
$L_K=\Inc\operatorname{diag}(K)\Inc^\top$ is the weighted Laplacian. A
\emph{synchronous} state solves
\begin{equation}\label{eq:sync}
  \Inc\big(K\circ\sin(\Inc^\top\theta)\big)=p .
\end{equation}
Write $\Delta=\Inc^\top\theta$, so $\Delta_e=\theta_i-\theta_j$ is the angle
difference \emph{across line} $e=(i,j)$ --- a relative quantity, never an absolute
bus angle --- and $F_e=K_e\sin\Delta_e$ is the active power flow on that line.
Call a state \emph{strictly cohesive} if $|\Delta_e|<\pi/2$ for every energized edge.  For the energy argument we also need the closed cohesive polytope
\begin{equation}\label{eq:omega}
  \Omega=\{\theta\in\R^n:\ |\Delta_e|\le\pi/2\ \ \forall e\in E\}.
\end{equation}
Thus a strictly cohesive state lies in $\intr\Omega$, whereas $\partial\Omega$ is the artificial barrier used by the direct-method certificate.  Throughout, $\intr S$ denotes the interior of a set $S$ and $\partial S$ its boundary.

Because the bus phases live on a torus, the $\Delta_e$ used to define a winding cell are principal edge differences, and their signed sum around a cycle can equal a nonzero multiple of $2\pi$.  Let $\Cyc\in\Z^{L\times c}$ be an oriented integer cycle-basis matrix whose columns represent a basis of directed graph cycles, so $\Inc\Cyc=0$ and $c=L-n+1$.  The winding coordinates of a cohesive torus state are then
$\varpi=\tfrac1{2\pi}\Cyc^\top\Delta\in\Z^c$.  Fixing any $F_p$ with $\Inc F_p=p$, every conservative flow is uniquely $F=F_p+\Cyc\psi$, $\psi\in\R^c$.
When we use the convex set $\Omega\subset\R^n$ below, we work in a fixed real lift of the zero-winding cell (and modulo the uniform phase); nonzero winding cells are represented instead by the linear tilt in the cycle-space formulation.

The swing dynamics, with inertia $M_i>0$ and damping $d_i>0$, are
\begin{equation}\label{eq:swing}
  M\dot\omega=p-d\circ\omega-\Inc\big(K\circ\sin(\Inc^\top\theta)\big),
  \qquad \dot\theta=\omega .
\end{equation}
Linearising \eqref{eq:swing} about a synchronous state gives the mass--spring
system with stiffness matrix $\Inc\operatorname{diag}(K\cos\Delta)\Inc^\top$, the
\emph{metagraph Laplacian}. Cohesion makes every metagraph weight
$K_e\cos\Delta_e$ nonnegative, so that matrix is positive semidefinite and the
state is stable: \emph{every cohesive synchronous state is a stable equilibrium of
\eqref{eq:swing}}.

The converse fails, and the failure is quantitative rather than incidental. Some
line angles may exceed $\pi/2$, making individual metagraph weights negative,
while the metagraph Laplacian retains a positive second eigenvalue and the state
remains stable. So $\Omega$ is a strict subset of the true stability region ---
Tavora's principal region \cite{tavora1972}, which $\Omega$ is inscribed in and
touches only at quadrature points. \Cref{sec:num} measures the resulting gap: on
IEEE 118 the last stable state has $\max_e|\Delta_e|=108.25^\circ$ with three of
$179$ lines past $90^\circ$, which is $3.7\%$ of loading beyond the last cohesive
state. Any criterion phrased as ``$|\Delta_e|\le\gamma<\pi/2$ for all $e$'' is
therefore sufficient and not necessary, and this paper is about making such a
criterion exact for what it claims, not about removing that restriction.

\section{Cycle-Space Convexity: Exact Cohesion Within a Winding Cell}\label{sec:dual}

The static problem becomes transparent after replacing bus phases by conservative line flows.  The integrated inverse-sine cost below is the convex conjugate of the cohesive swing potential.  In cycle coordinates its gradient is exactly the vector of cycle sums of principal angle differences, so stationarity enforces the required winding condition.

\begin{definition}\label{def:phi}
Let $g(s)=s\arcsin s+\sqrt{1-s^2}$ on $[-1,1]$ and
\begin{equation}\label{eq:phi}
\begin{aligned}
  \Phi(\psi)&=\sum_{e}K_e\,g(s_e),\\
  s_e&=\frac{(F_p+\Cyc\psi)_e}{K_e},\\
  \Omega_\psi&=\{\psi:|s_e|\le1\ \forall e\}.
\end{aligned}
\end{equation}
where $\Omega_\psi$ is a compact convex polytope, compact because $\Cyc$ has full
column rank.
\end{definition}

\begin{proposition}[Strict convexity]\label{prop:dual}
On $\intr\Omega_\psi$, with $\Delta_e=\arcsin s_e$,
\begin{equation}\label{eq:gradhess}
\begin{aligned}
  \nabla\Phi&=\Cyc^\top\arcsin(s),\\
  \nabla^2\Phi&=\Cyc^\top
  \operatorname{diag}(K_e\cos\Delta_e)^{-1}\Cyc\succ0 .
\end{aligned}
\end{equation}

\end{proposition}

\begin{proof}
$g'(s)=\arcsin s$, $g''(s)=(1-s^2)^{-1/2}>0$; differentiate and use
$K_e\sqrt{1-s_e^2}=K_e\cos\Delta_e$. Positive definiteness follows from
$\rank\Cyc=c$.
\end{proof}

The gradient $\nabla\Phi=\Cyc^\top\arcsin(s)$ is the vector of cycle sums of the principal angle map.  It becomes $2\pi$ times an integer winding vector only at a feasible stationary point of the corresponding tilted problem.  Existence in a fixed winding cell is therefore a stationarity question for a strictly convex function.

\paragraph{What a winding number means physically.}
Fix injections $p$.  Any two synchronous flow vectors satisfy the same conservation law $\Inc F=p$, so their difference lies in $\ker\Inc$ and is therefore a graph circulation.  The \emph{magnitude} of that circulation is not quantized.  What is quantized is the total principal phase advance around an oriented cycle:
$\Cyc_k^\top\Delta=2\pi\varpi_k$.  Thus distinct windings label distinct nonlinear equilibria with the same nodal injections but different cycle-space flow components \cite{coletta2016topologically,korsak1972}.

This topological use of ``loop flow'' should not be conflated with the operational power-system term for unscheduled parallel-path transfers, which can occur within the ordinary zero-winding branch.  Within the strictly cohesive set, two states with different winding vectors cannot be joined continuously while keeping every $|\Delta_e|<\pi/2$; a transition between winding cells must leave that set.  In the present paper $\varpi=0$ is not asserted to be universally ``the state an operator intends.''  It is the specific cell continuously connected to the small-angle/DC operating branch and is therefore the natural cell for the DCB comparison and for the transient experiment below.

\begin{theorem}[Cycle-coordinate realization of a winding cell]\label{thm:bij}
For $m\in\Z^c$ let $h_m(\psi)=\Phi(\psi)-2\pi m^\top\psi$, strictly convex on
$\Omega_\psi$.  Strictly cohesive solutions of \eqref{eq:sync} with winding vector $m$,
modulo a global phase, are in bijection with interior stationary points (equivalently,
interior minimizers) of $h_m$ on $\Omega_\psi$.  Hence there is at most one such state
per winding vector, consistently with the general winding-cell uniqueness theorem of
\cite{jafarpour2022torus}.  A state in the chosen cell exists exactly when the minimizer
lies in $\intr\Omega_\psi$.
\end{theorem}

\begin{proof}
If a strictly cohesive state exists, its conservative flow can be written uniquely as
$F=F_p+\Cyc\psi$ and the principal differences satisfy
$\nabla h_m(\psi)=\Cyc^\top\Delta-2\pi m=0$.  Strict convexity makes this stationary
point the unique minimizer in that cell.  Conversely, an interior minimizer satisfies
$\Cyc^\top\arcsin(s^\star)=2\pi m$; Kirchhoff's angle law then integrates the principal
differences to nodal phases on the torus, unique modulo the uniform phase, and flow
conservation gives \eqref{eq:sync}.  This is a cycle-coordinate derivation of the same
winding-cell structure established more generally in \cite{jafarpour2022torus}.
\end{proof}

One does not have to solve infinitely many programs. Cohesion bounds the winding
numbers, and the bound is purely topological.

\begin{lemma}[Strict-cohesion winding bound]\label{lem:finite}
Let cycle $k$ of the chosen basis have length $\ell_k$.  Every strictly cohesive state satisfies
\begin{equation}\label{eq:mbound}
  |\varpi_k| < \frac{\ell_k}{4},
  \qquad\text{hence}\qquad
  |\varpi_k|\le \Big\lfloor\frac{\ell_k-1}{4}\Big\rfloor .
\end{equation}
Consequently the rectangular candidate box contains at most
$\prod_k\big(2\lfloor(\ell_k-1)/4\rfloor+1\big)$ integer winding vectors.  In particular,
a nonzero $|\varpi_k|=q$ requires at least $4q+1$ edges in that basis cycle.
\end{lemma}

\begin{proof}
For a strictly cohesive state,
$2\pi|\varpi_k|\le\sum_{e\in k}|\Delta_e|<\ell_k\pi/2$, so
$|\varpi_k|<\ell_k/4$.  Taking the largest integer strictly below $\ell_k/4$ gives
$\lfloor(\ell_k-1)/4\rfloor$.
\end{proof}

Eq.~(\ref{eq:mbound}) is only a topological candidate bound; feasibility still has to be decided cell by cell.  The finiteness of the winding partition and all-solutions algorithms are already treated in \cite{jafarpour2022torus}; here the bound is used only to make the cycle-coordinate program explicit for the sinusoidal specialization.  For a single ring it reproduces the familiar strict-cohesion count $2\lfloor(n-1)/4\rfloor+1$ \cite{delabays2016multistability,manik2017cycle}.

\begin{proposition}[Where $\Phi$ comes from]\label{prop:fenchel}
Let $f_e(\Delta)=-K_e\cos\Delta$, so the swing potential is
\begin{equation}\label{eq:pot}
  \Vpot(\theta)=-p^\top\theta+\sum_ef_e\big((\Inc^\top\theta)_e\big).
\end{equation}
Each $f_e$ is convex on $(-\pi/2,\pi/2)$, hence $\Vpot$ is convex on the cohesive
polytope $\Omega$ of \eqref{eq:omega}, and
$f_e^*(F)=F\arcsin(F/K_e)+\sqrt{K_e^2-F^2}=K_eg(F/K_e)$.  More precisely, the Fenchel dual of minimizing $\Vpot$ is the concave maximization of $-\sum_e f_e^*(F_e)$ subject to $\Inc F=p$; equivalently, after eliminating conservation with $F=F_p+\Cyc\psi$, it is the convex minimization of $\Phi(\psi)=\sum_e f_e^*(F_e)$.

\end{proposition}

\noindent
The separable flow objective in \Cref{prop:fenchel} is not new as a convexification of the lossless synchronization equations: Bent, Bienstock and Chertkov wrote the corresponding line-flow program explicitly in 2013 \cite{bent2013syncaware}.  The present use of $\Cyc$ merely eliminates the conservation constraints and exposes the cycle coordinates and winding tilt.  In optimization language this is a monotropic/separable convex-cost network-flow structure \cite{rockafellar1984network}.

\begin{remark}\label{rem:whydual}
The primal potential \eqref{eq:pot} is convex on $\Omega$, but an ordinary Newton--Raphson power-flow iteration is a root solver for the nonlinear balance equations; it should not be described as implicitly minimizing the convex potential.  One may instead solve the primal convex program directly in the zero-winding cell.  The cycle formulation is useful because conservation is eliminated and different winding cells appear as the linear tilts $-2\pi m^\top\psi$ of the same separable strictly convex cost.
\end{remark}

\begin{remark}[Relation to Tavora's conjectures]\label{rem:tavconj}
\Cref{thm:bij} decides membership of $p$ in the image $f(\Omega)$ of the cohesive
polytope, by a convex program. It does \emph{not} prove either convexity
conjecture of \cite{tavora1972}, which concern the larger principal region and its
image. Whether $f(\Omega)$ is itself convex is open. What is immediate is a convex
outer bound: any cohesive state has $|F_e|\le K_e$, so
$f(\Omega)\subseteq\{\Inc F:|F|\le K\}$, a zonotope. The obstruction to convexity
is the constraint $\Cyc^\top\arcsin(F/K)\in2\pi\Z^c$, which is exactly what
\Cref{thm:bij} disposes of by fixing the winding vector first.
\end{remark}

\section{From the D\"orfler--Chertkov--Bullo Screen to Exact Cell-Wise Cohesion}\label{sec:improve}

The relation to the 2013 synchronization condition is especially simple: the widely used screen and the exact zero-winding oracle minimize two different costs over the same affine flow space.  One cost is quadratic; the other retains the full integrated inverse-sine nonlinearity.  This viewpoint turns the comparison from a competition between unrelated criteria into an approximation hierarchy.

The reference point is \cite{dorfler2013synchronization}: if
$\|L_K^\dagger p\|_{E,\infty}\le\sin\gamma$ for some $\gamma<\pi/2$, where
$\|x\|_{E,\infty}=\max_{\{i,j\}\in E}|x_i-x_j|$, then a unique stable cohesive
state with $|\Delta_e|\le\gamma$ exists. We can now say exactly what that
condition is.

\begin{proposition}[It is the quadratic Taylor polynomial of $\Phi$]\label{prop:taylor}
$g(s)=1+\tfrac{s^2}2+\tfrac{s^4}{24}+O(s^6)$. Truncating $\Phi$ after the
quadratic term gives, up to the constant $\sum_eK_e$,
$\Phi_2(\psi)=\sum_eF_e^2/(2K_e)$, whose minimiser over $\{\Inc F=p\}$ is
Thomson's principle, i.e.\ the DC power flow
$F_{\rm DC}=\operatorname{diag}(K)\Inc^\top L_K^\dagger p$; and
$\max_e|F_{{\rm DC},e}|/K_e=\|L_K^\dagger p\|_{E,\infty}$. So
\cite{dorfler2013synchronization} is the box test $|F_e|\le K_e\sin\gamma$
evaluated at the minimiser of the quadratic Taylor polynomial of $\Phi$, where the
exact test evaluates it at the minimiser of $\Phi$.
\end{proposition}

\begin{proof}
$g(0)=1$, $g'(0)=0$, $g''(0)=1$, $g'''(0)=0$, $g''''(0)=1$. Minimising $\Phi_2$
subject to $\Inc F=p$ is least squares in the metric
$\operatorname{diag}(K)^{-1}$, solved by the current distribution of the resistive
network with conductances $K_e$.
\end{proof}

This isolates precisely the distinction relevant to the \emph{selected zero-winding cell}.  First, the nonlinear map from flow to angle is $\arcsin$, and \cite{dorfler2013synchronization} already treats the target angle bound exactly through $\sin\gamma$; there is no small-angle replacement $\sin\gamma\mapsto\gamma$ in the stated DCB screen.  Second, the two constructions test different points of the affine flow space: DCB uses $F_{\rm DC}$, the minimizer of the quadratic surrogate $\Phi_2$, whereas the zero-winding synchronous solution, when it exists cohesively, is determined by $F^\star$, the minimizer of the full $\Phi$.  Under a small-loading scaling $p\mapsto\varepsilon p$, the first nonlinear correction to the cycle coordinate appears at cubic order in $\varepsilon$; on a tree there is no cycle coordinate and the two flows coincide identically.

Accordingly, the improvement claimed here is cell-specific: for a fixed winding vector, replacing the quadratic surrogate by the full strictly convex cost turns the DCB-style sufficient screen into an exact strict-cohesion test for that cell.  For the zero-winding cell this costs one strictly convex program in $c=L-n+1$ variables instead of one linear solve.  Deciding whether \emph{any} cohesive equilibrium exists across all windings still requires checking the finite set of candidate winding cells from \Cref{lem:finite}; that all-solutions problem is prior art \cite{jafarpour2022torus}.  Three corollaries make the zero-winding difference predictable before it is measured.

\begin{corollary}[Trees]\label{cor:tree}
If $G$ is acyclic then $c=0$, $\Omega_\psi$ is a point, $F_{\rm DC}=F^\star$
identically, and the 2013 condition is exact with no approximation of any kind.
\end{corollary}

This is a one-line proof of tree-exactness, and more usefully it identifies the
cycle dimension $c=L-n+1$ as the precise measure of what is lost off trees.

A \emph{bridge} (or cut edge) is an edge whose removal disconnects $G$;
equivalently, an edge belonging to no cycle; equivalently, an edge whose row in any
cycle-basis matrix $\Cyc$ is identically zero. A bridge is a cut-set with a single
element.

\begin{corollary}[Bridges]\label{cor:bridge}
If $e$ is a bridge then row $e$ of $\Cyc$ is zero, so $F_e=F_{p,e}$ for every
$\psi$: no cycle flow can relieve a bridge. The exact test and the 2013 condition
therefore agree exactly whenever the binding constraint sits on a bridge, and can
differ only when it sits on a cycle-carrying edge.
\end{corollary}

\begin{corollary}[A binding bridge makes the Jacobian singular]\label{cor:bridgefold}
Suppose at the strict-cohesion boundary the only saturated edge $e$ is a bridge, so
$|\Delta_e|=\pi/2$ while $|\Delta_{e'}|<\pi/2$ for $e'\ne e$.  Then the metagraph
weight $K_e\cos\Delta_e$ vanishes; removing that positive-weight connection disconnects
the metagraph, and the stiffness/Jacobian acquires a second zero eigenvalue.  Thus the
strict-cohesion boundary is a Jacobian singularity in this case.  Calling that point a
\emph{saddle-node} additionally requires the usual nondegeneracy/transversality conditions,
which are not proved here.
\end{corollary}

In the language of the classical power-flow geometry \cite{tavora1972}, a bridge reaching
quadrature produces a singular point on the boundary of the cohesive polytope.  
The MATPOWER case39 numerics below show precisely this Jacobian singularity; we do not infer a generic bifurcation type from singularity alone.

\begin{remark}[No coordinate-wise ordering follows from the objective ordering]\label{rem:dir}
Pointwise, $g(s)-1\ge s^2/2$, so the nonlinear separable objective lies above its quadratic
surrogate.  This fact alone does \emph{not} imply an ordering of the corresponding minimizers
in $\ell_\infty$, nor does it prove that every line flow moves away from saturation.  The DCB
criterion is safe because it is a proved sufficient condition \cite{dorfler2013synchronization},
not because of a coordinate-wise comparison between $F_{\rm DC}$ and $F^\star$.
\end{remark}

\Cref{cor:bridge,cor:bridgefold} are the sharpest falsifiable claims here, and
\Cref{sec:num} tests them on two systems that fall on opposite sides.

\section{From Equilibrium Existence to a Transient Energy Certificate}\label{sec:trans}

The static oracle answers whether a prescribed post-event operating point exists in the selected cohesive winding cell.  It does not synthesize a response policy and it does not prove that the physical trajectory reaches that point.  The transient layer therefore takes the balancing allocation as prescribed data and asks a separate question: whether the pre-event state lies inside a provably invariant post-event energy sublevel.

Existence is necessary and not sufficient. We now certify that the machines
actually reach the new equilibrium.

\paragraph{The disturbance actually simulated.} The proof-of-principle experiment below does \emph{not} model a raw load rejection followed in time by primary-frequency response.  Instead, at $t=0$ the injection vector is changed instantaneously and in a balanced way,
\begin{equation}\label{eq:peff}
  p_{\rm eff}=p+\Delta P\,(e_\delta-\gamma),
  \qquad \gamma\ge0,\quad \mathbf1^\top\gamma=1.
\end{equation}
The first term removes a load block at bus $\delta$; the second term applies the balancing generation reduction \emph{at the same instant}. Eq.~\eqref{eq:peff} should be read as the zero-delay limit of a \emph{pre-provisioned} balancing scheme.  The vector $\gamma$ specifies the spatial participation pattern that droop, governor, converter or other fast controls are intended eventually to realize; here their finite response time is collapsed to zero so that the balancing action is present at the instant the load disappears.  This is an optimistic idealization: it removes the bulk active-power mismatch and therefore suppresses the frequency excursion that would normally trigger primary response.  It is nevertheless nontrivial because the spatial injection pattern changes discontinuously.  Even with perfect instantaneous balancing in aggregate, the resulting rotor-angle transient can cross the cohesive boundary.  Failure of the certificate under this favorable model is therefore a strong warning; success is only a statement about the idealized limit and not a guarantee for a staged droop/governor response. Thus $\mathbf1^\top p_{\rm eff}=0$ and there is no modeled bulk-frequency excursion, governor lag, dead band or droop trajectory.  The vector $\gamma$ is best read as a spatial balancing-allocation parameter.  Varying $\gamma$ while holding the old operating point and $\Delta P$ fixed isolates a geometric question: how strongly does the transient certificate depend on where the compensating injection change is placed?

This idealization is intentionally optimistic and is not a claim about the temporal behavior of a data-center protection system or generator controls.  A physically staged load-loss/governor model is a natural extension, but it is outside the experiment reported here.

\paragraph{The certificate.} Let $\theta^\star$ be the strictly cohesive equilibrium for $p_{\rm eff}$ and let $\Vpot$
be \eqref{eq:pot} with $p$ replaced by $p_{\rm eff}$.  After the instantaneous balanced
step the swing dynamics are simply \eqref{eq:swing} with injections $p_{\rm eff}$ and
initial state $(\theta_{\rm old},\omega=0)$.  Hence
\begin{align}
  W(\theta,\omega)
  &=\tfrac12\omega^\top M\omega
    +\Vpot(\theta)-\Vpot(\theta^\star),\label{eq:lyap}\\
  \dot W&=-\omega^\top\!\operatorname{diag}(d)\,\omega\le0 .
\end{align}
and at the instant of the step
\begin{equation}\label{eq:W0}
  W_0=\Vpot(\theta_{\rm old})-\Vpot(\theta^\star).
\end{equation}
No center-of-inertia subtraction or steady-frequency offset is needed for the experiment as actually implemented.

Because $\mathbf1^\top p_{\rm eff}=0$, the potential and the equilibrium equations are invariant under a uniform phase shift.  The stability statement below is therefore understood on the phase quotient $\R^n/\operatorname{span}\{\mathbf1\}$ (equivalently, after fixing one angle when evaluating the static energy landscape).  Convergence of $\theta$ means convergence modulo this uniform phase.

\begin{theorem}[Transient certificate]\label{thm:trans}
Let $\Omega$ be as in \eqref{eq:omega} with $\theta^\star\in\intr\Omega$, let
$d_i>0$, and put
\begin{equation}\label{eq:cstar}
  c^\star=\min\big\{\Vpot(\theta)-\Vpot(\theta^\star)\ :\ \theta\in\partial\Omega\big\}.
\end{equation}
If $W_0<c^\star$ then the solution of \eqref{eq:swing} from $(\theta_{\rm old},0)$
stays cohesive for all $t\ge0$, satisfies $\omega(t)\to0$, and its relative angles converge to the equilibrium class $[\theta^\star]$ on the uniform-phase quotient.
\end{theorem}

\begin{proof}
By \Cref{prop:fenchel}, $\Vpot$ is convex on $\Omega$ and, after fixing the uniform-phase gauge, $\theta^\star$ is its unique critical point.  Hence $W\ge0$ on the quotient phase space, with equality only at the equilibrium class $([\theta^\star],0)$, and $\dot W\le0$. Suppose the trajectory first reaches
$\partial\Omega$ at time $t_1$. Then
$W(t_1)\ge\Vpot(\theta(t_1))-\Vpot(\theta^\star)\ge c^\star>W_0\ge W(t_1)$, a
contradiction; so the trajectory remains in $\intr\Omega$. There $W$ is a
Lyapunov function, 
and, on the phase quotient, the largest invariant set in $\{\dot W=0\}=\{\omega=0\}$ is the equilibrium class $[\theta^\star]$.  LaSalle's invariance principle therefore gives convergence of the relative angles to $\theta^\star$ and $\omega\to0$.

\end{proof}

The certificate does not involve the damping magnitude: only $d\succ0$ enters, via
the sign of $\dot W$. Damping moves the true threshold but never the certified
one, which makes the bound robust to the least reliable parameter in the model.

The classical difficulty with direct methods is computing $c^\star$: it is usually
approached by searching for a ``controlling unstable equilibrium'', a nonconvex
problem with no guarantee that the right one was found
\cite{chiang1995direct,athay1979practical}. For one and two loops the difficulty
can be sidestepped by projecting the Lyapunov function onto the plane of the loop
angles and reading minima and saddles off a contour plot, as in
\cite{coletta2016topologically}; that is exact in two dimensions and does not
survive to a meshed network with $c=62$ independent cycles. Convexity removes the
search without removing the guarantee.

\begin{theorem}[The critical energy is a finite family of convex programs]\label{thm:cstar}
$\partial\Omega=\bigcup_{e=1}^{L}\bigcup_{s=\pm1}F_{e,s}$ with
$F_{e,s}=\Omega\cap\{\Delta_e=s\pi/2\}$. Each $F_{e,s}$ is a compact convex set on
which $\Vpot$ is convex, so
\begin{equation}
  c^\star=\min_{e,s}\ \Big[\min_{\theta\in F_{e,s}}\Vpot(\theta)\Big]-\Vpot(\theta^\star)
\end{equation}
is the minimum over at most $2L$ nonempty signed-edge face problems (empty or redundant faces may be ignored).  The face decomposition is mathematically exact; each numerical face minimum is computed only to the tolerance of the chosen convex solver.
\end{theorem}

\begin{proof}
$\Omega$ is the polytope $\{|\Delta_e|\le\pi/2\}$, whose boundary is the union of
its facets, each obtained by activating one constraint; intersecting a convex set
with a hyperplane preserves convexity, and $\Vpot$ is convex on $\Omega$ by
\Cref{prop:fenchel}. Compactness holds after fixing the gauge $\theta_1=0$, since
connectivity bounds $|\theta_i|$ by $(\pi/2)\operatorname{diam}(G)$.
\end{proof}

\begin{remark}[What is certified]\label{rem:whatcert}
It is the certified energy sublevel
$\{(\theta,\omega):\theta\in\Omega,\ W(\theta,\omega)<c^\star\}$---not all of $\Omega$---that is contained in the region of attraction of $[\theta^\star]$.  On that sublevel, \Cref{thm:trans} certifies something stronger than eventual convergence: the trajectory never leaves the cohesive set, so no energized-line angle reaches $\pm\pi/2$ even transiently.  \Cref{sec:num} compares this deliberately conservative invariant-sublevel guarantee with finite-horizon swing simulation.
\end{remark}

\paragraph{What this section delivers, and to whom.}
For the idealized balanced-step model, the output is a screening number: the largest $\Delta P$ at a specified bus and a specified balancing allocation $\gamma$ for which the trajectory is
guaranteed to remain in $\Omega$ and converge.  The calculation uses network couplings, the old operating point and the chosen balanced post-step injection vector.  In planning it can compare
candidate sites and spatial balancing rules; in operations it can be recomputed as the dispatch changes.  It is not a substitute for electromagnetic-transient or governor/protection studies.

\section{From Theory to an Engineering Screening Workflow}
\label{sec:workflow}

A review is useful only if the theoretical objects can be translated into a decision sequence.  The intended use here is hierarchical screening.  The calculation does not replace full AC, electromagnetic-transient, protection or detailed governor studies.  It supplies a transparent angle-stability layer that can be inserted before those more expensive analyses and that can explain \emph{why} a case is easy, difficult, or topology-limited.

\subsection{Inputs, outputs, and escalation}

The static inputs are the energized graph, line couplings $K_e$, a balanced post-event injection vector $p$, and---when multiple winding cells are relevant---the cell to be tested.  The first calculation can remain the D\"orfler--Chertkov--Bullo screen because it requires only the weighted-Laplacian pseudoinverse.  If that sufficient condition passes, there is no reason to solve a stronger oracle merely to reproduce the same guarantee.  If it fails, the cycle-space program provides the natural second stage: it decides exact strict-cohesion feasibility in the selected winding cell.  On networks with many bridges, the bridge decomposition is useful before either solve because a bridge flow is fixed by conservation and cannot be relieved by any cycle correction.

For a transient query one additionally specifies the pre-event phase state, inertia and positive damping, together with the \emph{prescribed} post-event balancing allocation.  The post-event cohesive equilibrium becomes the center of the energy function.  The critical boundary energy is then the minimum of independent signed-edge face problems.  These face minimizations can be distributed across processors because no face depends on the solution of another.  If the initial energy lies below the resulting boundary minimum, the trajectory is certified to remain cohesive and converge.  If it does not, the method returns ``not certified,'' not ``unstable''; the appropriate escalation is nonlinear simulation or a less conservative stability analysis.

\begin{table*}[t]
\centering
\caption{A practical hierarchy for using the certificate.  Each stage is invoked only if the cheaper preceding stage is insufficient for the engineering question.}
\label{tab:workflow}
\small
\begin{tabular}{p{0.12\textwidth}p{0.20\textwidth}p{0.22\textwidth}p{0.32\textwidth}}
\toprule
stage & calculation & positive result means & negative or inconclusive result means \\
\midrule
1 & D\"orfler--Chertkov--Bullo Laplacian screen & a cohesive stable equilibrium is guaranteed & no conclusion; proceed to the nonlinear cycle-space oracle \\
2 & exact convex cycle-space solve in a chosen winding cell & that cell contains a strictly cohesive equilibrium & that cell is infeasible; inspect other admissible windings only if the application requires them \\
3 & convex boundary-energy calculation & the specified pre-event state lies in an invariant cohesive energy sublevel & no conclusion about instability; the certificate may simply be conservative \\
4 & nonlinear time simulation or higher-fidelity model & trajectory-level evidence under the chosen model & model- and horizon-dependent; still not a mathematical certificate for the real grid \\
\bottomrule
\end{tabular}
\end{table*}

\subsection{Planning, operations, and control are different uses}

In planning, the framework can compare candidate data-center locations, contingency sizes, network reinforcements and alternative spatial participation patterns.  The bridge corollary is especially interpretable in this setting: if the limiting transfer is forced through a bridge, nonlinear loop redistribution cannot improve the cohesive margin, so a more sophisticated cycle-space solve will not change the bottleneck.  In a meshed area, by contrast, the gap between the quadratic and nonlinear flows measures how much the network cycles can redistribute stress before a line reaches quadrature.

In operations, the same calculations can be repeated around a changing dispatch because the graph and cycle basis are largely reusable while $p$ changes.  The method is therefore naturally suited to contingency screening and sensitivity studies.  It is not, in its present form, an online emergency controller.  The vector $\gamma$ used below specifies a balancing action supplied to the model.  Comparing several choices of $\gamma$ answers a verification question---which prescribed allocation has the larger certified margin?---but it does not synthesize the best admissible feedback law.  Keeping this distinction explicit prevents the transient experiment from being mistaken for a control-design result.

Finally, the hierarchy makes the scope of ``exact'' precise.  The cycle-space solve is exact for strict cohesion \emph{inside a fixed winding cell} of the lossless fixed-voltage model.  The boundary-energy decomposition is exact for the corresponding direct-method energy barrier, up to numerical solver tolerance.  Neither statement makes the reduced model exact for an AC transmission grid, nor does either identify the complete nonlinear basin of attraction.  This separation between mathematical exactness and modeling fidelity is essential when a certificate is used in an engineering workflow.

\section{Open-Source Benchmark Validation}\label{sec:num}

The benchmarks serve two roles appropriate to a review-and-synthesis article.  First, small synthetic tests verify the analytical identities independently of the power-system cases.  Second, standard MATPOWER networks show when the integrated certificate materially changes the engineering conclusion and when topology forces it to agree with the simpler screen.

All results below are deterministic and intentionally synthetic.  The network topologies and static data are drawn from MATPOWER case39 (the 39-bus New England system) and the MATPOWER IEEE 118-bus case \cite{zimmerman2011matpower}; they are reduced to lossless fixed-voltage swing models.  We avoid calling case39 an ``IEEE-39'' test system: MATPOWER describes it as a New England 39-bus model that is only generally representative of the underlying 345-kV system.  These cases are useful for proof-of-principle algorithmic tests, but neither should be read as a calibrated dynamic model of an actual contemporary transmission system.  Synthetic inertias, damping and balanced-step controls are used only to expose the certificate's structure.

The companion code implements the cycle-space objective, the DCB comparison, equilibrium continuation, face-wise critical-energy minimization and finite-horizon swing simulation.

\subsection{The dual is the right object}\label{sec:numdual}

Three checks, in increasing strength.

\emph{Derivatives and reconstruction.}  On the random-graph validation instance, finite differences give maximum gradient and Hessian discrepancies of approximately $3.0\times10^{-9}$ and $1.7\times10^{-10}$, respectively.  Reconstructing nodal angles from the cycle-space minimizer gives an angle-consistency residual of order $3\times10^{-14}$.

\emph{Trees.}  On random trees the exact strict-cohesion threshold and the DCB threshold coincide to numerical precision, as \Cref{cor:tree} requires.

\emph{Winding counts.}  For rings of $n=3,\ldots,25$ at zero injection, the cycle-coordinate programs reproduce the strict-cohesion count $2\lfloor(n-1)/4\rfloor+1$.  The strict bound in \Cref{lem:finite} implies that $|\varpi|=q$ needs at least $4q+1$ edges.  For the MATPOWER case39 basis used in the code, with cycle lengths $3,4,4,5,7,7,10,17$, the per-cycle bounds are $0,0,0,1,1,1,2,4$, so the rectangular candidate box contains $1{,}215$ winding vectors.

The static comparison is more informative in words than as a bar plot because the two test cases illustrate different mechanisms.  In MATPOWER case39 the DCB and exact cohesive thresholds coincide: the first binding branch is a bridge, so its flow has no cycle-space degree of freedom that the nonlinear solution could use to redistribute stress.  In case118 the binding edge lies in the meshed part of the network; the nonlinear cycle-space minimizer can redistribute flow relative to the quadratic/DC minimizer, producing a visible gap between the sufficient DCB screen and the exact zero-winding cohesion threshold.  The numerical continuation beyond that threshold then measures a different gap---between strict cohesion and local stability---and should not be conflated with the DCB-versus-exact comparison.

\subsection{Existence: how much the exact test buys, and when it buys nothing}
\label{sec:numexist}

\Cref{tab:brack} compares the DCB sufficient threshold with the exact strict-cohesion threshold and with a warm-started numerical continuation of a locally stable branch.  The third number is included only as a numerical reference for what happens beyond $\Omega$; it is not promoted to a certified bifurcation point.

\begin{table*}[t]\centering
\caption{Loading multiplier $\alpha$: DCB sufficient threshold, exact strict-cohesion threshold, and a numerically continued locally stable branch limit.}
\label{tab:brack}
\small
\begin{tabular}{lccccl}
\toprule
system & $c=L-n+1$ & $\alpha_{\rm DCB}$ & $\alpha_{\rm coh}$ & $\alpha_{\rm cont}$ & cohesive binding line \\
\midrule
IEEE 118       & $62$ & $4.363697$ & $5.068736$ $(+16.16\%)$ & $5.268367$ & on a cycle \\
case39 & $8$ & $5.544474$ & $5.544472$ $(\approx0\%)$ & $5.544472$ & bridge $6$--$31$ \\
\bottomrule
\end{tabular}
\end{table*}

On IEEE 118 the exact strict-cohesion oracle permits $16.16\%$ more loading than the DCB sufficient condition.  Relative to the numerically continued branch limit, it closes $77.93\%$ of the interval between the DCB and continuation thresholds.  At the strict-cohesion boundary the relevant Jacobian is still nonsingular, with $\lambda_2\approx0.205$, so there is no basis for calling that point a saddle-node.  Continuing the locally stable branch beyond the cohesive boundary eventually produces states with several line differences exceeding $90^\circ$; the endpoint reported in Table~\ref{tab:brack} is therefore a numerical continuation result, not an exact stability boundary.

On MATPOWER case39 the DCB and strict-cohesion thresholds coincide because the binding branch is the bridge between buses $6$ and $31$.  The bridge row of the cycle basis is zero, so no cycle flow can relieve it.  When that bridge reaches quadrature its metagraph weight vanishes and the Jacobian acquires an additional zero mode, exactly as \Cref{cor:bridgefold} predicts.  This establishes singularity; a generic saddle-node classification would require additional nondegeneracy checks.

In IEEE 118 the cohesive-binding edge is not a bridge.  Farther along the continued locally stable branch, three lines have crossed $90^\circ$.  These observations describe two distinct geometric facts---meshed redistribution at the cohesive boundary and local stability beyond cohesion---and no theorem is claimed that relates the number of beyond-quadrature lines to the width of the continuation interval.

\subsection{The transient: where the balancing action is placed}\label{sec:numtrans}

The transient proof-of-principle uses the MATPOWER case39 reduction at $\alpha=4$.  The selected load block is at MATPOWER bus~20 and the fully concentrated balancing action is at bus~38.  The pre-step maximum energized-line angle difference is approximately $46.17^\circ$.  These are deliberately stressed, synthetic conditions chosen so the certificate binds; they are not offered as representative operating statistics for a real system.

It is useful to calibrate what ``large angle'' can mean in a genuinely severe event without confusing unlike quantities.  In the FERC/NERC reconstruction of the September~8, 2011 Southwest blackout, the phase-angle separation \emph{between the two terminals of an out-of-service 500-kV line} increased from about $20^\circ$ to approximately $72^\circ$ \cite{fercnerc2012blackout}.  That $72^\circ$ number is not an energized-branch angle and should not be compared literally with the line differences in our lossless model; it simply illustrates that tens of degrees occur in extreme post-contingency conditions.  Routine small-signal operation is not the regime targeted by the experiment below.

\Cref{tab:trans} compares, at $\alpha=4$, the energy-certified threshold, the finite-horizon simulated cohesion threshold, and the static strict-cohesion feasibility limit for five spatial balancing allocations.  ``Simulated'' means that the  Fourth Order Runge-Kutta (RK4) integration remained inside $|\Delta_e|<\pi/2$ over the stated finite horizon; it is not a theorem about asymptotic stability or the exact basin boundary.

\begin{table*}[t]\centering
\caption{Thresholds [MW] for the synthetic balanced injection step on the MATPOWER case39 reduction at $\alpha=4$, with a $2724.5$~MW load block at bus~20.  Column~$0$ spreads the simultaneous balancing action over generators; column~$1$ concentrates it at bus~38.}
\label{tab:trans}
\small
\begin{tabular}{lccccc}
\toprule
concentration of balancing action & $0$ & $0.25$ & $0.5$ & $0.75$ & $1$ \\
\midrule
energy-certified                  & $2035$ & $1751$ & $1442$ & $1186$ & $991$ \\
finite-horizon simulated cohesion & $2724$ & $2724$ & $2585$ & $2442$ & $2313$ \\
exact static strict-cohesion      & $2724$ & $2724$ & $2724$ & $2724$ & $2724$ \\
DCB static screen                 & $2724$ & $2724$ & $2724$ & $2724$ & $2724$ \\
\midrule
certified / simulated             & $0.747$ & $0.643$ & $0.558$ & $0.486$ & $0.428$ \\
\bottomrule
\end{tabular}
\end{table*}

The static tests accept the complete balanced step in every allocation, while the energy certificate is strongly spatially selective.  Moving from distributed balancing to a single remote balancing bus lowers the certified step from about $2035$ to $991$~MW, even though the load block and the total balanced injection change are otherwise the same.  The finite-horizon simulation also becomes less tolerant as the balancing action is concentrated, but it remains substantially less conservative than the certificate.

That conservatism is the intended direction.  The certificate guarantees that the trajectory never reaches the artificial boundary $|\Delta_e|=\pi/2$; it does not attempt to characterize the full nonlinear region of attraction.  The finite-horizon simulations are diagnostics, not the standard against which a proof should be tuned.
\begin{figure*}[t]\centering
  \includegraphics[width=0.58\textwidth]{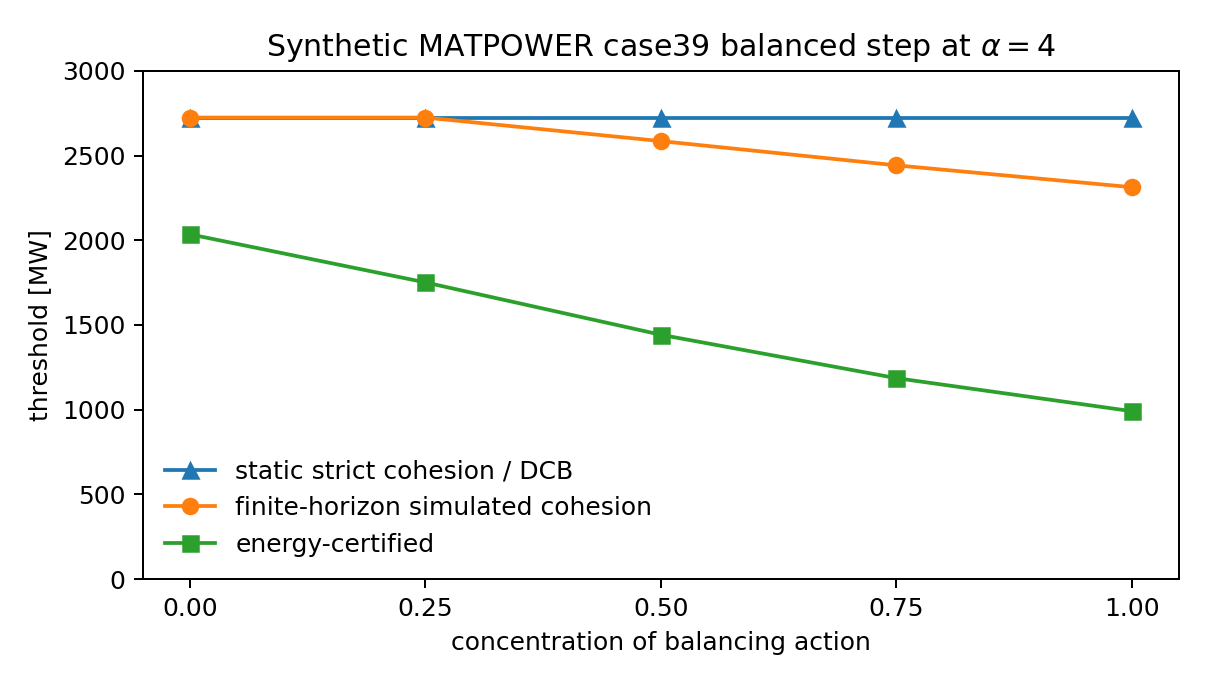}
\caption{Synthetic MATPOWER case39 balanced-step thresholds from Table~\ref{tab:trans}.  The exact static strict-cohesion and D\"orfler--Chertkov--Bullo screens accept the full $2724.5$~MW step for all allocations.  The finite-horizon swing simulation remains cohesive to a lower threshold when the balancing action is concentrated, while the energy certificate is more conservative and falls from about $2035$ to $991$~MW.  The separation between the static and transient curves is the practical reason to keep equilibrium feasibility and transient security as distinct layers of the certificate.}
  \label{fig:transient}
\end{figure*}

\begin{remark}\label{rem:coincidence}
In the MATPOWER case39 calculation the critical-energy face is associated with the same bridge,
buses $6$--$31$, that binds the static strict-cohesion limit.  This is an empirical observation for this test case, not a theorem that the steady and transient limiting faces must coincide.
\end{remark}

\section{Engineering Interpretation, Limitations, and Outlook}\label{sec:disc}

For a broad engineering audience, the main lesson is structural.  The same network geometry that classifies nonlinear equilibria also organizes a computable transient certificate.  The price of this unification is a deliberately restricted model and a deliberately conservative safety region; the benefit is that every step of the argument has a clear guarantee and a clear failure mode.

\paragraph{What the convexity bought.} One object served both questions. For
existence, within a fixed winding cell, it converted a sufficient quadratic-surrogate screen into an exact convex test and explained the
DCB/2013 condition as a quadratic truncation, with tree-exactness, the bridge case and
the bridge-induced Jacobian singularity of \Cref{cor:bridgefold} falling out as corollaries rather than as separate arguments. For the transient it converted the critical-energy computation from a nonconvex search into at most $2L$ convex signed-edge face programs.  Neither step required anything beyond \Cref{prop:fenchel}: the swing potential is convex on the cohesive polytope.

\paragraph{Interpretation of the proof-of-principle experiment.}
The numerical result says that, within this idealized model, the spatial allocation of a \emph{simultaneous balanced injection change} can alter the energy margin by roughly a factor of two.  It should not be read as a claim about fixed droop gain or fixed steady-frequency offset, because those dynamics are not present in the implemented experiment.  The broader literature on data-center flexibility and runtime grid coordination \cite{khanal2026shift,chung2026openg2g} makes spatial/temporal allocation a practically relevant design variable; a physically staged governor/load-control model is needed before translating the present certificate into procurement guidance.

\paragraph{Limitations, stated plainly.}
The study uses lossless, fixed-voltage network reductions and synthetic inertia/damping parameters.  The disturbance is an instantaneous \emph{balanced} injection change; governor lags, dead bands, converter controls, voltage/reactive-power dynamics, dynamic loads, transient overloads and protection actions are absent.  These omitted layers can interact with the angle dynamics and can dominate a real event.  The energy certificate is conservative by construction because it insists on never leaving the closed cohesive polytope.  The finite-horizon simulation thresholds are not exact basin boundaries.  Finally, although the strict winding bound makes the candidate set finite, exhaustive enumeration of all winding cells on a large meshed network is not demonstrated here; the static numerical comparison is performed on the zero-winding branch.

\paragraph{Why revisit this problem now?}
There is also a methodological --- and, for the author, persona l--- reason for the paper.  Roughly a decade ago I worked on direct synchronization-stability metrics for transmission-grid contingencies \cite{dorfler2013synchronization,bent2013syncaware,backhaus2014syncmetrics}.  At the time, I viewed loss of synchrony as a real but increasingly specialized risk: demand response, distributed resources, microgrids and faster controls appeared to be making the classical ``large coherent disturbance'' picture less central.  The rapid emergence of concentrated computational facilities during the recent year changes that assessment.  A single customer class can now create a multi-gigawatt load step on a time scale relevant to system dynamics.  That does not invalidate the older stability machinery; it makes it useful again under a new forcing pattern.

This perspective also explains the paper's form.  It is partly a review of tools developed in several communities --- synchronization conditions, winding-cell geometry, separable convex network flow and direct transient-energy methods --- and partly an attempt to recombine them around the new large-load problem.  The novelty claimed here is therefore  integrative: the cycle-space convex oracle, its precise relation to the DCB sufficient screen, and the finite family of convex boundary-energy problems are assembled into one static-to-dynamic certificate rather than presented as unrelated results.

\paragraph{Path forward.}

The next layer is to relax the idealizations one at a time rather than to enlarge the present certificate by accretion.  Finite-delay load loss and primary response would turn the simultaneous balanced step into a staged trajectory; stochastic generation and demand would turn deterministic cohesion into a probabilistic first-exit question; and feedback-capable resources would change the task from verifying a prescribed balancing allocation to synthesizing an implementable control policy with a quantitative certificate.  These are natural extensions of the deterministic geometry developed here, but they are distinct problems and are deliberately left outside the present article.  A complete engineering treatment must also reconnect the angle layer to voltage/reactive-power dynamics, thermal overloads, converter controls and protection actions.

\paragraph*{Code availability.}
A companion Python/Jupyter package supplied with the manuscript reproduces the cycle-space derivative checks, winding-count tests, MATPOWER case39/case118 static calculations, continuation experiments, critical-energy face minimizations, transient integrations, and the reported figure and tables.  The numerical package is intended to make every benchmark claim in the article independently reproducible from standard open-source network data.

\paragraph*{Disclaimer on the use of language models.} A large language model was
used to proof-read the manuscript and to help design, implement and check the
numerical experiments. All results were reviewed by the author, who is solely
responsible for the content, the claims and any errors.

\bibliographystyle{IEEEtranN}
\bibliography{refs}

@article{dorfler2013synchronization,
  author  = {D\"orfler, Florian and Chertkov, Michael and Bullo, Francesco},
  title   = {Synchronization in complex oscillator networks and smart grids},
  journal = {Proceedings of the National Academy of Sciences},
  volume  = {110}, number = {6}, pages = {2005--2010}, year = {2013},
  doi     = {10.1073/pnas.1212134110}
}

@article{manik2017cycle,
  author  = {Manik, Debsankha and Timme, Marc and Witthaut, Dirk},
  title   = {Cycle flows and multistability in oscillatory networks},
  journal = {Chaos},
  volume  = {27}, number = {8}, pages = {083123}, year = {2017},
  doi     = {10.1063/1.4994177}
}

@article{korsak1972,
  author  = {Korsak, Andrew J.},
  title   = {On the question of uniqueness of stable load-flow solutions},
  journal = {IEEE Transactions on Power Apparatus and Systems},
  volume  = {PAS-91}, number = {3}, pages = {1093--1100}, year = {1972},
  doi     = {10.1109/TPAS.1972.293463},
  note    = {Stanford Research Institute. Gives the first explicit counterexample
             to uniqueness of stable load-flow solutions, and observes that it
             ``could easily occur in the `donut' of the Western United States
             interconnected system''.}
}

@article{tavora1972,
  author  = {Tavora, Carlos J. and Smith, Otto J. M.},
  title   = {Equilibrium analysis of power systems},
  journal = {IEEE Transactions on Power Apparatus and Systems},
  volume  = {PAS-91}, number = {3}, pages = {1131--1137}, year = {1972},
  doi     = {10.1109/TPAS.1972.293469}
}

@article{coletta2016topologically,
  author  = {Coletta, Tommaso and Delabays, Robin and Adagideli, In\c{c}i and
             Jacquod, Philippe},
  title   = {Topologically protected loop flows in high voltage AC power grids},
  journal = {New Journal of Physics},
  volume  = {18}, pages = {103042}, year = {2016},
  doi     = {10.1088/1367-2630/18/10/103042}
}

@article{delabays2016multistability,
  author  = {Delabays, Robin and Coletta, Tommaso and Jacquod, Philippe},
  title   = {Multistability of phase-locking and topological winding numbers in
             locally coupled Kuramoto models on single-loop networks},
  journal = {Journal of Mathematical Physics},
  volume  = {57}, number = {3}, pages = {032701}, year = {2016},
  doi     = {10.1063/1.4943296}
}

@article{chiang1995direct,
  author  = {Chiang, Hsiao-Dong and Chu, Chia-Chi and Cauley, Gerry},
  title   = {Direct stability analysis of electric power systems using energy
             functions: theory, applications, and perspective},
  journal = {Proceedings of the IEEE},
  volume  = {83}, number = {11}, pages = {1497--1529}, year = {1995},
  doi     = {10.1109/5.481632}
}

@article{athay1979practical,
  author  = {Athay, T. and Podmore, R. and Virmani, S.},
  title   = {A practical method for the direct analysis of transient stability},
  journal = {IEEE Transactions on Power Apparatus and Systems},
  volume  = {PAS-98}, number = {2}, pages = {573--584}, year = {1979},
  doi     = {10.1109/TPAS.1979.319407}
}

@book{rockafellar1984network,
  author    = {Rockafellar, R. Tyrrell},
  title     = {Network Flows and Monotropic Optimization},
  publisher = {Wiley}, year = {1984}
}

@article{zimmerman2011matpower,
  author  = {Zimmerman, Ray D. and Murillo-S\'anchez, Carlos E. and Thomas, Robert J.},
  title   = {{MATPOWER}: Steady-state operations, planning, and analysis tools for
             power systems research and education},
  journal = {IEEE Transactions on Power Systems},
  volume  = {26}, number = {1}, pages = {12--19}, year = {2011}
}

@article{jafarpour2022torus,
  author  = {Jafarpour, Saber and Huang, Elizabeth Y. and Smith, Kevin D. and Bullo, Francesco},
  title   = {Flow and Elastic Networks on the {$n$}-Torus: Geometry, Analysis, and Computation},
  journal = {SIAM Review},
  volume  = {64},
  number  = {1},
  pages   = {59--104},
  year    = {2022},
  doi     = {10.1137/18M1242056}
}

@inproceedings{bent2013syncaware,
  author    = {Bent, Russell and Bienstock, Daniel and Chertkov, Michael},
  title     = {Synchronization-Aware and Algorithm-Efficient Chance Constrained Optimal Power Flow},
  booktitle = {2013 IREP Symposium on Bulk Power System Dynamics and Control -- IX Optimization, Security and Control of the Emerging Power Grid},
  year      = {2013},
  doi       = {10.1109/IREP.2013.6629400},
  eprint    = {1306.2972},
  archivePrefix = {arXiv}
}

@misc{pjm2026dominion,
  author       = {{PJM Interconnection}},
  title        = {July 22, 2026 Dominion Load Transfer Event},
  howpublished = {System Operations Subcommittee presentation, July 31, 2026},
  year         = {2026},
  note         = {Reports approximately 3,800 MW of load transferring to backup power after a 230-kV fault; ACE reached +3,928 MW and frequency 60.092 Hz}
}

@INPROCEEDINGS{zhao2014intermittent,
	author={Zhao, Changhong and Chertkov, Michael and Backhaus, Scott},
	booktitle={2015 48th Hawaii International Conference on System Sciences}, 
	title={Optimal Sizing of Voltage Control Devices for Distribution Circuit with Intermittent Load}, 
	year={2015},
	pages={2680-2689},
	doi={10.1109/HICSS.2015.323}}

@article{khanal2026shift,
  author       = {Khanal, Saroj and Roh, Geon and Yao, Boyu and Silverman, Abraham and Gayme, Dennice and Konstantinou, Charalambos and Kim, Jip and Dvorkin, Yury},
  title        = {Shift or curtail? How much data-center flexibility is worth depends on the host power grid},
  year         = {2026},
 journal     = {arXiv:2608.19622},
  archivePrefix= {arXiv},
  primaryClass = {physics.soc-ph}
}

@article{chung2026openg2g,
  author       = {Chung, Jae-Won and Liang, Zhirui and Mao, Yanyong and Chen, Jiasi and Chowdhury, Mosharaf and Dvorkin, Vladimir},
  title        = {{OpenG2G}: A Simulation Platform for AI Datacenter-Grid Runtime Coordination},
  year         = {2026},
  journal       = {arXiv:2605.05519},
  archivePrefix= {arXiv},
  primaryClass = {cs.LG}
}

@techreport{fercnerc2012blackout,
  author      = {{Federal Energy Regulatory Commission} and {North American Electric Reliability Corporation}},
  title       = {Arizona--Southern California Outages on September 8, 2011: Causes and Recommendations},
  institution = {FERC and NERC},
  year        = {2012},
  month       = apr,
  note        = {Staff report; reconstruction reports the phase-angle separation between the terminals of the out-of-service Hassayampa--North Gila 500-kV line increasing from about 20 degrees to about 72 degrees}
}

@article{almada2026dynamic,
	author        = {Ayrton Almada and Laurent Pagnier and Igal Goldshtein and Saif R. Kazi and Michael Chertkov},
	title         = {Real-Time Dynamic N-1 Screening: Identifying High-Risk Lines and Transformers After Common Faults},
	year          = {2026},
	journal        = {arXiv:2602.12293},
	archivePrefix = {arXiv},
	primaryClass  = {math.OC}
}

@article{backhaus2014syncmetrics,
	author        = {Scott Backhaus and Russell Bent and Daniel Bienstock and Michael Chertkov and Dvijotham Krishnamurthy},
	title         = {Efficient Synchronization Stability Metrics for Fault Clearing},
	year          = {2014},
	journal       = {arXiv:1409.4451},
	archivePrefix = {arXiv},
	primaryClass  = {cs.SY}
}

@misc{pjm2026largeload,
	author       = {{PJM Interconnection}},
	title        = {{PJM, Dominion Review Large Load Transfer Event}},
	howpublished = {PJM Inside Lines},
	month        = aug,
	year         = {2026},
	note         = {Aug. 11, 2026. Available: \url{https://insidelines.pjm.com/pjm-dominion-review-large-load-transfer-event/}}
}

@misc{ferc2026computational,
	author       = {{Federal Energy Regulatory Commission}},
	title        = {{Order Directing the North American Electric Reliability Corporation to File Reliability Standard(s) Pertaining to Computational Load Integration}},
	howpublished = {Docket No. RD26-7-000},
	month        = jul,
	year         = {2026},
	note         = {July 16, 2026. Summary available: \url{https://www.ferc.gov/news-events/news/summaries-july-2026-commission-meeting}}
}

@misc{nerc2026largeloads,
	author       = {{North American Electric Reliability Corporation}},
	title        = {{Large Loads Action Plan: Addressing Reliability Risks of Emerging Large Loads}},
	year         = {2026},
	note         = {Large Loads Working Group resources and 2026 action-plan updates. Available: \url{https://prod.nerc.com/initiatives/large-loads-action-plan}}
}

@techreport{nerc2026risk,
	author      = {{North American Electric Reliability Corporation}},
	title       = {{Risk Mitigation for Emerging Large Loads}},
	institution = {North American Electric Reliability Corporation},
	year        = {2026},
	month       = may,
	note        = {Reliability Guideline. Available: \url{https://www.nerc.com/globalassets/our-work/guidelines/reliability/RG_Risk-Mitigation-For-Emerging-Large-Loads.pdf}}
}

\begin{comment}
\begin{IEEEbiography}
	[{\includegraphics[width=1in,height=1.25in,clip,keepaspectratio]{Misha09_17_24-headshot.pdf}}]
	{Michael (Misha) Chertkov}
	is Professor of Mathematics at the University of Arizona. He received his Ph.D. in physics from the Weizmann Institute of Science in 1996, was an R.~H. Dicke Fellow at Princeton University, and joined Los Alamos National Laboratory as a J.~R. Oppenheimer Fellow in 1999. After two decades at Los Alamos, he moved to the University of Arizona in 2019. He is a Fellow of the American Association for the Advancement of Science (AAAS) and the American Physical Society (APS), and a Senior Member of IEEE.
	
	His research lies at the interface of applied mathematics, statistical physics, control, optimization, and machine learning, with applications to complex natural and engineered systems. His work on energy systems includes power-grid synchronization and stability, optimization and control under uncertainty, rare-event dynamics, and data-driven and generative methods for grid modeling and decision making. More broadly, his current research explores stochastic dynamics, inference, generative artificial intelligence, and control.
\end{IEEEbiography}
\end{comment}

\end{document}